\documentclass[10pt,a4paper]{amsproc}

\usepackage{latexsym, amsthm}

\usepackage{amsfonts, amsmath, amssymb}

\usepackage{euscript,mathrsfs}

\newtheorem{theorem}{Theorem}[section]
\newtheorem{lemma}[theorem]{Lemma}
\newtheorem{proposition}[theorem]{Proposition}

\newtheorem{corollary}[theorem]{Corollary}

\theoremstyle{definition}
\newtheorem{definition}[theorem]{Definition}

\newtheorem{question}[theorem]{Question}

\theoremstyle{remark}
\newtheorem{remark}[theorem]{Remark}
\def\Fq{{\mathbb F}_q}

\def\g{{\gamma}}

\def\M{{\mathsf M}}

\def\FP{{\mathsf {FP}}}

\def\Sd{{S_{\le d}}}
\def\A{{\mathcal A}}

\def\hp3{\widehat{\mathbb P}^3}
\newcommand{\supp}{\mathrm{Supp}}
\newcommand{\Jmon}{\mathcal{J}_{\mathsf{Mon}}}

\newcommand{\ev}{\mathrm{ev}}

\newcommand{\Sp}{\mathrm{Span}}

\def\HF{\mathsf{HF}}

\def\Z{{\mathsf Z}}

\def\ZZ{{\mathbb Z}}
\def\Fu2{{F_{u_1}^{u_2}}}

\def\LT{\mathsf{LT}}

\begin{document}
\title[Relative Generalized Hamming weights of affine cartesian codes]{Relative Generalized Hamming weights of affine Cartesian codes}
\author{Mrinmoy Datta}
\thanks{The author is supported by a postdoctoral fellowship from DST-RCN grant INT/NOR/RCN/ICT/P-03/2018.}
\address{Institute of Mathematics and Statistics, \newline \indent
University of Troms\o \\
Troms\o, Norway}
\email{mrinmoy.dat@gmail.com}

\begin{abstract}
We explicitly determine all the relative generalized Hamming weights of affine Cartesian codes using the notion of footprints and results from extremal combinatorics. This generalizes the previous works on the determination of relative generalized Hamming weights of Reed-Muller codes by Geil and Martin, as well as the determination of all the generalized Hamming weights of the affine Cartesian codes by Beelen and Datta. 
\end{abstract}

\maketitle

\section{Introduction}
Determination of parameters of Reed-Muller type codes have received a lot of attention from several mathematicians in recent past. In this paper, we look at a certain class of codes, called the affine Cartesian codes, that comes naturally as a generalization of Reed-Muller Codes. These codes were introduced in 2013 by Geil and Thomsen \cite{GT1} in a more general setting of weighted Reed-Muller codes. The name ``affine Cartesian codes" was coined by L\'opez, Renter\'ia-M\'arquez and Villarreal \cite{LRV} in 2014. Since then several articles have appeared where the parameters of these codes were studied extensively. Like in the case of Reed-Muller codes, the problem of computing parameters such as minimum distance, generalized Hamming weights etc., of affine Cartesian codes translates to the problem of determination of the maximum number of common zeroes of systems of polynomials satisfying certain properties in a subset of an affine space over a finite field. The fundamental properties of affine Cartesian codes, such as their dimensions and the minimum distances, were obtained in \cite{LRV}. Later in 2018, the generalized Hamming weights \cite{BD} of the affine Cartesian codes were completely determined. This generalizes the classical work \cite{HP} of Heijnen and Pelikaan towards the determination of all the generalized Hamming weights of the Reed-Muller codes. Several articles, for example \cite{C, CN}, are devoted towards the determination of the next to minimal weights of affine Cartesian codes. 

The notion of generalized Hamming weights of a code was introduced by Wei \cite{W} in 1991  in order to characterize the code performance of on a wire tap channel of type II. A generalization of this notion, known as generalized Hamming weights or higher weights, was defined and studied by Wei \cite{W} in 1991. A generalization of these weights is known as the relative generalized Hamming weight of a code $C_1$ with respect to a proper subcode $C_2$. This notion was introduced by Luo, Mitrpant, Han Vinck and Chen \cite{Lea}, again towards studying new characters on the wire tap channel of type II, in 2005 and was further studied in a subsequent article \cite{LCL} by Liu, Chen and Luo. For the definition of the relative generalized Hamming weights of linear codes we refer to Subsection \ref{sec:rel}.

As the title of the article indicates, we are interested in determining the relative generalized Hamming weights of an affine Cartesian codes with respect to a subcode which is again an affine Cartesian codes.  This work generalizes the result in the article \cite{GM} where the authors have determined all the relative generalized Hamming weights of the Reed-Muller codes. Also, the main results of the current article can be viewed as a generalization of the result in \cite{BD} which gives all the generalized  Hamming weights of affine Cartesian codes. In proving our result in this paper, we follow the footsteps of \cite{GM} and \cite{BD}, where the results were derived using the notion of so-called footprint bound. Some early articles on  footprint bounds include \cite{FL, H, GH} and some recent articles include \cite{NPV, GMVV, BDG} among others. A somewhat brief discussion of the notion of the footprint bounds is given in Subsection \ref{sec:fp}.

The paper is organized as follows. In Section 2, we recall most of the definitions and the known results that will be used in proving our main theorem. In Section 3, we deduce Theorem \ref{w}, which can be viewed as an extension of the famous Kruskal-Katona Theorem in extremal combinatorics. Finally, in Section 4, we state and prove the main result of the paper where we compute all the relative generalized Hamming weights of an affine Cartesian code with respect to a smaller affine Cartesian code.

\section{Preliminaries}
We devote this section to recalling the well-known definitions and results that will be used in the sequel. In particular, we recall the definitions of relative generalized Hamming weights of a code with respect to a smaller subcode and the notion of affine Cartesian codes in the following two subsections. Later, we revisit the notion of the so called footprint bound which helps us in translating the algebraic geometric problem of determination of the maximum number of common zeroes of  certain systems of polynomials in a specified subset of the affine space over a projective space into a seemingly different problem in extremal combinatorics. We will conclude this section by introducing some combinatorial notations which will be used in the next section. In particular, none of the results or definitions mentioned in this section are new. For a thorough understanding of the results that are mentioned here a reader is encouraged to see the references mentioned and the references therein. 

\subsection{Relative generalized Hamming weights of linear codes} \label{sec:rel}
We begin this subsection by recalling the definition of the relative generalized Hamming weights of a code with respect to a proper subcode. Throughout, we will denote by $\Fq$ a finite field with $q$ elements where $q$ is a prime power. 

\begin{definition}\cite[Definition 2]{LCL}
Let $C_2 \subsetneq C_1$ be linear codes and $\ell := \dim C_1 - \dim C_2$. For $r = 1, \dots, \ell$, the $r$-th \textit{relative generalized Hamming weights}  of $C_1$ with respect to $C_2$ (RGHW of $C_1$ w.r.t. $C_2$) is defined as
$$M_r (C_1, C_2) := \min_{J \subseteq \{1, \dots, n\}} \{ |J| : \dim ((C_1)_J) - \dim ((C_2)_J) = r \},$$
where $(C_i)_J = \{c = (c_1, \dots, c_n) \in C_i \mid c_t =0 \ \text{for} \ t \notin J\}$ for $i = 1, 2$. The sequence $(M_1 (C_1, C_2), \dots, M_{\ell} (C_1, C_2))$ is known as the hierarchy of RGHWs of $C_1$ w.r.t. $C_2$.
\end{definition}

The following Lemma, which can be found as \cite[Lemma 1]{LCL}, gives an alternative definition of the RGHWs of a code $C_1$ w.r.t. a proper subcode $C_2$. 

\begin{lemma}\label{lem:defn}\cite[Lemma 1]{LCL}
Let $C_2 \subsetneq C_1$ be linear codes and $\ell = \dim C_1 - \dim C_2$. For $r = 1, \dots, \ell$, we have
\begin{equation}\label{defn}
M_r (C_1, C_2) = \min \left\{|\supp (D)| : D \subset C_1;  D \cap C_2 = \{0\}, \dim D = r\right\},
\end{equation}
where, given a subspace $D$ of $\Fq^n$, the support of $D$, denoted by $\supp (D)$, is given by 
$$\supp (D) := \left\{i \in \{1, \dots, n\} \mid c_i \neq 0 \ \text{for some} \ (c_1, \dots, c_n) \in D \right\}.$$
\end{lemma}

In  what follows, we will use the equation \eqref{defn} as our definition of the RGHWs. 
\begin{remark}
In view Lemma \ref{lem:defn}, it is clear that  if $C_2 = \{0\}$, then the RGHWs of $C_1$ w.r.t. $C_2$ are exactly the generalized Hamming weights of $C_1$. 
\end{remark}

\subsection{Affine Cartesian codes}\label{sec:affcar}
In this subsection, we recall the definition of the affine Cartesian codes. Throughout, we will use the convention that the degree of the zero polynomial is $-1$.

\begin{definition}
Let $d_1 \le \dots \le d_m$ be positive integers and $A_1, \dots, A_m$ are subsets of $\Fq$ with cardinalities $d_1, \dots, d_m$ respectively. Denote by $\mathcal{A}$ the cartesian product $\mathcal{A}:=A_1 \times \cdots \times A_m$. Note that $|\mathcal{A}| = n:= d_1 \cdots d_m$. Further, fix an enumeration $P_1, \dots, P_n$ of elements in $\mathcal{A}$ and a positive integer $d \le k:=\sum_{i=1}^m (d_i - 1).$ For $d \le k$, define the subspace
$$S_{\le d}{(\mathcal{A})} := \{f \in \Fq[x_1, \dots, x_m] : \deg_{x_i} f \le d_i -1 \ \mathrm{and} \ \deg f \le d\}.$$
The map 
$$\ev: S_{\le k} (\mathcal{A}) \to \Fq^{|\mathcal A|} \ \ \ \mathrm{by} \ \ \ \ f \mapsto (f(P_1), \dots, f(P_n))$$
is a linear map and consequently, for each $d \le k$, the image $AC_q (d, \mathcal{A}) := \ev (\Sd (\mathcal{A}))$ is a linear subspace of $\Fq^n$ and is called \textit{affine cartesian codes}.
\end{definition}

Henceforth, we will write $A_i := \{\g_{i, 1}, \dots, \g_{i, d_i}\}$ for $i = 1, \dots, m$. It is not hard to show  that the map $\ev$ is one-one. This implies that the dimension of $AC_q (d, \A)$ is same as $\dim S_{\le d}(\mathcal{A})$.  As mentioned in the introduction, we are interested in the determination of the RGHWs of an affine Cartesian code w.r.t. a ``smaller" affine Cartesian code. More precisely, our goal is to answer the following:

\begin{question}\label{q1}
Let $u_1, u_2$ be integers satisfying $-1 \le u_2 < u_1 \le k$. Determine  $M_r (AC_q (u_1, \mathcal{A}), AC_q (u_2, \mathcal{A}))$, for $r \le  \dim AC_q (u_1, \mathcal{A}) - \dim AC_q (u_2, \mathcal{A})$.
\end{question}

For ease of  notations, we will denote $M_r(u_1, u_2):=M_r (AC_q (u_1, \mathcal{A}), AC_q (u_2, \mathcal{A}))$ and $\ell:=\dim AC_q (u_1, \mathcal{A}) - \dim AC_q (u_2, \mathcal{A})$. We note that if $u_2 = -1$, then $M_r (u_1, u_2)$ are simply the $r$-th generalized Hamming weights of $AC_q (u_1, \mathcal{A})$. In the recent work \cite{BD}, the generalized Hamming weights of affine Cartesian codes were completely determined. To answer the above question we introduce the following sets. For an integer $r \le \ell$, we define,
$$\mathcal{D}_r := \{D \subset AC_q (u_1, \mathcal{A}) \mid D \cap AC_q (u_2, \mathcal{A}) = 0; \dim D = r\}.$$
We endow the set of monomials in $\Fq[x_1, \dots, x_m]$ with the graded lexicographic order. In the following Lemma we give a necessary and sufficient condition for a subspace of $AC_q (u_1, \mathcal{A})$ to be a member of $\mathcal{D}_r$.

\begin{lemma}\label{gen}
Let $D$ be a subspace of $AC_q (u_1, \mathcal{A})$ of dimension $r$. Then $D \in \mathcal{D}_r$ iff there exists $f_1, \dots, f_r \in S_{\le d}(\A)$ with $D = \Sp \{\ev (f_1), \dots, \ev(f_r)\}$ satisfying the following three conditions:
\begin{enumerate}
\item[(C1)] $f_1, \dots, f_r$ are linearly independent,
\item[(C2)] $u_2 < \deg \LT (f_i) \le u_1$ for $i = 1, \dots, r$,
\item[(C3)] $\LT (f_i) \neq \LT (f_j)$ whenever $i \neq j$.
\end{enumerate}
Consequently, $|\supp (D)| = n - |\Z_{\mathcal{A}} (f_1, \dots, f_r)|,$
where $\Z_\A (f_1, \dots, f_r)$ denotes the set of common zeroes of $f_1, \dots, f_r \in \A$.
\end{lemma}

\begin{proof}
It is easy to see that the three conditions are sufficient. To see that they are also necessary, we begin with  $D \in \mathcal{D}_r$, and a set of $r$ linearly independent polynomials $f_1, \dots, f_r$ such that $D = \Sp \{\ev (f_1), \dots, \ev(f_r)\}$. It is clear that the polynomial $f_1$ satisfies the condition (C2). For  $2 \le k \le r$, we replace $f_k$ by a linear combination of $f_1, \dots, f_k$ so that  the polynomials $f_1, \dots, f_k$ satisfy the condition (C3). Clearly the condition (C2) is satisfied for $f_1, \dots, f_r$. The last assertion follows trivially. 
\end{proof}
We now define the following family consisting of sets of $r$ polynomials:
$$\mathcal{C}_r := \{\{f_1, \dots, f_r \} \mid f_1, \dots, f_r \ \text{satisfy (C1), (C2), (C3)}\}$$
It follows directly from Lemma \ref{gen} that 
\begin{equation}\label{M}
M_r (u_1, u_2) 
= n - \max \{|\Z_{\mathcal{A}} (f_1, \dots, f_r)| : \{f_1, \dots, f_r \} \in \mathcal{C}_r\}.
\end{equation}
We have thus shown that the Question \ref{q1} is equivalent to the following question:
\begin{question}\label{q2}
For integers $r, u_1, u_2$ and the set $\A$ as above, determine 
$$a_r (u_1, u_2, \mathcal{A}) := \max \{|\Z_{\mathcal{A}} (f_1, \dots, f_r)| : \{f_1, \dots, f_r \} \in \mathcal{C}_r\}.$$
\end{question}

\subsection{The footprint bound}\label{sec:fp}
In order to answer Question \ref{q2} we will use the footprint bound. This method of producing upper bounds on generalized Hamming weights of Reed-Muller type codes is dependent on the theory of Gr\"obner bases and that of affine Hilbert functions. For a comprehensive reading on these notions, the reader is referred to \cite{CLO}. Most of what follows in this section can be found in \cite[Section 2]{BD}. We provide a somewhat detailed description of what will be used later for the sake of completeness and ease of readability.

Let us denote by $S$ the polynomial ring $\Fq[x_1, \dots, x_m]$ and for any integer $u$ we define $S_{\le u} :=\{f \in S  \mid \deg f \le u \}.$ For any ideal $I$ of $S$, we define $I_{\le u} := I \cap S_{\le u}$. The affine Hilbert function of $I$, denoted by $^a \HF_I$, is defined as
$$^a\HF_I : \ZZ \to \ZZ \ \ \ \ \text{given by} \ \ \ \ ^a\HF_I (u) := \dim S_{\le u} - \dim I_{\le u}.$$
It is easy to derive that if $I$ and $J$ are ideals of $S$ with $I \subset J$, then for any $u \in \ZZ$ we have $^a \HF_J(u) \le ^a \HF_I (u)$. For a subset $X \subset \Fq^m$ we define the ideal $I(X)$ to be the ideal of $S$ consisting of polynomials vanishing everywhere in $X$. For such a subset $X \subset \Fq^m$, we define its affine Hilbert function, denoted by $^a \HF_X$, as  $^a \HF_X := ^a\HF_I(X)$. 

\begin{proposition}\label{found}
\
\begin{enumerate} 
\item[(a)] \cite[Section 9.3]{CLO} Let $\prec$ be any graded order on $S$. Then 
\begin{enumerate}
\item[(i)] For any ideal $I$ of $S$, we have $^a\HF_{\LT (I)} (u) = ^a\HF_I (u)$.
\item[(ii)] If $I$ is a monomial ideal of $S$, then $^a\HF_I(u)$ is given by the number of monomials of degree at most $u$ that do not lie in $I$
\end{enumerate}
\item[(b)] \cite[Lemma 2.1]{NW} If $Y \subset \Fq^m$ is a finite set, then $|Y| = ^a \HF_Y (u)$ for all sufficiently large values of $u$.
\end{enumerate}
\end{proposition}

Similar statements as in the above proposition could also be found, albeit in disguise of footprints, in \cite[Corollary 4.5]{G} and in \cite[Corollary 2.5]{CLO2}. The above Proposition helps us in finding out an upper bound for the quantity $|\Z_{\mathcal{A}} (f_1, \dots, f_r)|$ for a given $\{f_1, \dots, f_r \} \in \mathcal{C}_r$. To this end, we see that the polynomials $g_1, \dots, g_m \in I(Z_{\A} (f_1, \dots, f_r))$, where 
$$g_j := \prod_{k=1}^{d_j} (x_j - \g_{j, k}) \ \ \ \ \text{for} \ \ \ \ j = 1, \dots, m.$$ 
At this juncture, it will be useful to assign some notations for the ideals in question. Define
$$\mathcal{I} := I(Z_{\A} (f_1, \dots, f_r)) \ \ \ \text{and} \ \ \ \LT (\mathcal{I}):= \ \text{the leading term ideal of} \ I.$$
Furthermore, we have the monomial ideals:
$$\mathcal{J} := \langle f_1, \dots, f_r, g_1, \dots, g_m \rangle \ \ \ \text{and}  \ \ \ \ \Jmon:= \langle \LT (f_1), \dots, \LT (f_r), x_1^{d_1}, \dots, x_m^{d_m}\rangle.$$
It follows trivially from the above discussions that, $\mathcal{J} \subset \mathcal{I}$ and that 
\begin{equation}\label{eqone}
\Jmon \subseteq \LT (\mathcal{J}) \subseteq \LT (\mathcal{I})
\end{equation}
Using Proposition \ref{found} and equation \eqref{eqone} we see that for sufficiently large $u$,
\begin{equation}\label{eqtwo}
|\Z_{\A} (f_1, \dots, f_r)| = ^a\HF_{\mathcal{I}} (u) = ^a\HF_{\LT (\mathcal{I})} (u) \le ^a\HF_{\Jmon} (u)  
\end{equation}

Let us write $\M = \{\mu \in S \ | \  \mu \ \text{is a monomial} \}$. It follows from from Proposition \ref{found} (a) (ii) that 
$$^a\HF_{\Jmon} (u) = |\{\mu \in \M : \deg \mu \le u, x_i^{d_i} \nmid \mu, \LT (f_j) \nmid \mu \ \text{for} \ i = 1, \dots, m \ \text{and} \ j=1, \dots, r\}|.$$
Furthermore, if we take $u \ge \sum_{i=1}^m d_i$, then 
$$^a\HF_{\Jmon} (u) = |\{\mu \in \M : \deg_{x_i} \mu \le d_i - 1, \LT (f_j) \nmid \mu \ \text{for} \ i = 1, \dots, m \ \text{and} \ j=1, \dots, r\}|.$$
We define $$\M_{\A} := \{\mu \in \M \mid \deg_{x_i} \mu \le d_i -1 \ \text{for} \ i=1, \dots, m \},$$ 
and given any set of monomials $m_1, \dots, m_r$, the set of footprints,
$$\FP_\A (m_1, \dots, m_r) := \{\mu \in \M_\A : m_i \nmid \mu \ \text{for} \ i=1, \dots, r\}.$$
The previous discussions now imply that 
\begin{equation}\label{eqfp1}
|\Z_{\A} (f_1, \dots, f_r)| \le |\FP_\A (\LT (f_1), \dots, \LT (f_r))|.
\end{equation}

The upper bound on the number of points on $Z_{\A} (f_1, \dots, f_r)$ thus obtained from equation \eqref{eqfp1} is referred to as the footprint bound. Indeed,
\begin{equation}\label{eqmax}
a_r (u_1, u_2, \A) \le \max \left\{ |\FP_\A (\LT (f_1), \dots, \LT (f_r))| : \{f_1, \dots, f_r\} \in \mathcal{C}_r\right\}.
\end{equation}
In the following subsection, we will introduce some combinatorial notions which will help us in deriving the right hand side of the equation \eqref{eqmax}.

\subsection{Some combinatorial tools}
In this subsection, we will introduce some combinatorial notions that will help us in translating the problem of determining the right hand side of the equation \eqref{eqmax} to a problem of extremal combinatorics. Let
$$F = \{0,\dots, d_1 - 1\} \times \dots \times \{0, \dots, d_m - 1\}. $$
We have two natural orderings for the elements of $F$, namely the lexicographic order and the partial order. Let us write
$$(a_1, \dots, a_m) \prec_{lex} (b_1, \dots, b_m)$$
if $(a_1, \dots, a_m)$ is less than $(b_1, \dots, b_m)$ in lexicographic order, i.e. there exists $j$ with $1 \le j \le m$ such that $a_i =  b_i$ for all $i < j$ and $a_j < b_j$. Also we will write $$(a_1, \dots, a_m) \prec_{P} (b_1, \dots, b_m)$$ if and only if $(a_1, \dots, a_m)$ is less than $(b_1, \dots, b_m)$ in partial order, i.e. $a_i \le b_i$ for all $i = 1, \dots, m$ and for some $j \in \{1, \dots, n\}$ we have $a_j < b_j$.We write $(a_1, \dots, a_m) \preceq_{lex} (b_1, \dots, b_m)$ (resp. $(a_1, \dots, a_m) \preceq_{P} (b_1, \dots, b_m)$) if $(a_1, \dots, a_m) \prec_{lex} (b_1, \dots, b_m)$ (resp. $(a_1, \dots, a_m) \prec_{P} (b_1, \dots, b_m)$) or $(a_1, \dots, a_m) = (b_1, \dots, b_m)$. We have a bijection 
$$\phi : \M_\A \to F \ \ \text{given by} \ \ x_1^{a_1} \cdots x_m^{a_m} \mapsto (a_1, \dots, a_m).$$
It is clear that for $\mu_1, \mu_2 \in \M_\A$, we have $\mu_1 \mid \mu_2$ if and only if $\phi(\mu_1) \preceq_P \phi(\mu_2)$.
Now for $\textbf{a} = (a_1, \dots, a_m) \in F$, we define $\deg(\textbf{a}):=a_1+\cdots+a_m$. Let us introduce some subsets of $F$ consisting of elements satisfying certain degree constraints: for any integer $u$, define
$$F_u := \{\textbf{a} \in F : \deg(\textbf{a}) = u\} \ \ \ \text{and} \ \ \  F_{\le u} := \{\textbf{a} \in F : \deg(\textbf{a}) \le u\}. $$
On a similar note, for integers $u_1, u_2$ satisfying $u_2 < u_1$, we define 
$$F_{u_2}^{u_1} := \{\textbf{a} \in F : u_2 < \deg(\textbf{a}) \le u_1\}.$$
Given a subset $S \subset F$, we define the shadow (resp. footprint) of $S$ in $F$, denoted by $\nabla (S)$ (resp. $\Delta(S)$) as follows:
$$\nabla (S) := \{ \textbf{b} \in F \mid \textbf{a} \preceq_P \textbf{b} \ \text{for some} \ \textbf{a}  \in S\}  \ \ \text{and} \ \ \Delta(S) := F \setminus \nabla (S).$$
For an integer $u$, we define $\Delta_u (S) := \Delta (S) \cap F_u$ and $\nabla_u (S) := \nabla (S) \cap F_u$. It now follows from equation \eqref{eqmax} that 
\begin{equation}\label{eqfp}
a_r (u_1, u_2, \A) \le \max \{|\Delta (S)| : S \subset F_{u_2}^{u_1}, |S| = r\}.
\end{equation}
In the subsequent section, we will derive the exact value of the right hand side in the above inequality. Before concluding this section, we remark that the field $\Fq$ does not play an essential role as long as we are interested in computing the quantity $a_r (u_1, u_2, \A)$. The inequalities \eqref{eqmax} and \eqref{eqfp} continue to hold even if we replace $\Fq$ by an arbitrary field having at least $d_m$ elements.

\section{Result from Combinatorics}
Motivated from the discussion in the last section,  we now investigate the following question.
\begin{question}\label{q3}
Fix integers $u_1, u_2$ and $r$ with $-1 \le u_2 < u_1 \le k$. Denote by $\mathcal{F}_r$, the family of subsets of $F_{u_2}^{u_1}$ of cardinality $r$. Determine $\max \{|\Delta(S)| : S \in \mathcal{F}_r\}.$
\end{question}

We remark that if $d_1 = d_2 = \dots = d_m = q$, then the answer to this question is known in various cases: 
\begin{enumerate}
\item for $u_2 = -1$, this question corresponds to the determination of the GHWs of the Reed-Muller codes, which was solved by Heijnen and Pellikaan in \cite{HP}.
\item in general, without any constraint on $u_2$, the question corresponds to the determination of the RGHWs of the Reed-Muller codes, and as mentioned before, this question was answered by Geil and Martin in \cite{GM}. 
\end{enumerate}
Furthermore, in the general situation with $d_1 \le \dots \le d_m$, this problem was solved in \cite{BD} in the case $u_2 = -1$ in order to determine the GHWs of the affine Cartesian codes. In order to proceed, we first introduce the following two notations:
\begin{enumerate}
\item[(a)] For and integer $u$ and a subset $S \subset F_u$, we define $L(S)$ to be the set consisting of the first $|S|$ elements of $F_u$ in descending lexicographic order.
\item[(b)] For integers $u_1, u_2$ with $-1 \le u_2 < u_1 \le k $ and a subset $S \subset F_{u_2}^{u_1}$, we define $N(S)$ to be the set consisting of the first $|S|$ elements of $F_{u_2}^{u_1}$ in descending lexicographic order.
\end{enumerate}

The following classical Theorem, due to Clements and Lindstr\"om, will play an instrumental role in the sequel. 

\begin{theorem}\cite[Corollary 1]{CL}\label{CLT}
Let $u < k$ and $S \subseteq F_u$. Then $$\nabla_{u+1} (L(S)) \subseteq L (\nabla_{u+1} (S)).$$
\end{theorem}

The following is an easy corollary of the Theorem \ref{CLT}.  

\begin{corollary}\label{cl} 
For integers $u, v$ with $u \le v \le k$ and $S \subset F_u$, we have  
\begin{enumerate}
\item [(a)] \cite[Corollary 3.2]{BD} $\nabla_v(L(S)) \subseteq L (\nabla_v(S))$ and thus, $|\nabla_v (L(S))| \le |\nabla_v(S)|$.
\item[(b)] \cite[Corollary 3.3]{BD}  $|\nabla (L(S))| \le |\nabla (S)|$.
\end{enumerate}
\end{corollary}

In order to prove our main results, we will also need the following lemma that can be found in \cite[Lemma 3.5 and Remark 3.6]{BD}. 

\begin{lemma}\label{proplex}
Fix integers $u, v$ with $u < v \le k$ and an element $\textbf{y} \in F_v$. If  $\textbf{a}_\textbf{y}:= \max_{lex} \{\textbf{f} \in F_u : \textbf{f} \le_{lex} \textbf{y}\}$, then $\textbf{a}_\textbf{y} \preceq_P \textbf{y}$.
\end{lemma}

The following two lemmas are motivated from their analogues \cite[Lemma 3.6 and Lemma 3.7]{BD}. We include the proofs for the sake of completeness. 

\begin{lemma}\label{claim}
Let $u, u_1, u_2$ be integers satisfying $-1 \le u_2 < u \le u_1 \le k$. Let $N(r)$ denote the first $r$ elements of $F_{u_2}^{u_1}$ in descending lexicographic order. If $N_u := N(r) \cap F_u$ and $r_u := |N_u|$, then  $$\nabla_{u_1} (N_u) \subseteq N_{u_1} \subseteq \nabla_{u_1} (N_u^*), $$ where $N_u^*$ consists of the first $r_u+1$ elements of $F_u$ in descending lexicographic order.
\end{lemma}

\begin{proof}
The result is trivially true if $u = u_1$. So we may assume that $u < u_1$. Let $\textbf{y} \in \nabla_{u_1} (N_u)$. Then there exists $\textbf{x} \in N_u$ such that $\textbf{x} \preceq_P \textbf{y}$. Consequently $\textbf{x} \preceq_{lex} \textbf{y}$. Since $\textbf{x} \in N(r)$  and $\textbf{x} \preceq_{lex} \textbf{y}$, we have $\textbf{y} \in N(r)$. Since $\textbf{y} \in F_{u_1}$, we have  $\textbf{y} \in N(r) \cap F_{u_1} = N_{u_1}$. 

Now let $\textbf{y} \in N_{u_1}$. Define $\textbf{a}:= \max_{lex} \{\textbf{f} \in F_u : \textbf{f} \preceq_{lex} \textbf{y}\}$. From Lemma \ref{proplex}, we obtain $\textbf{a} \le_P \textbf{y}$. If $\textbf{a} \in N_u$, then $\textbf{a} \in N_{u^*}$, which proves the assertion. So we may assume that $\textbf{a} \not\in N_u$. Clearly, the set $N_u$ consists of the first $r_u$ elements of $F_u$ in descending lexicographic order. If we write $N_{u^*} =\{ \textbf{f}_1, \dots, \textbf{f}_{r_u + 1}\}$, then  $\textbf{a} \preceq_{lex} \textbf{f}_{r_u + 1}$.
If $\textbf{a}  = \textbf{f}_{r_u + 1}$, then  $\textbf{a} \in N_u^*$, and the assertion follows. Now suppose, if possible, that $\textbf{a} \prec_{lex} \textbf{f}_{r_u + 1}$. The maximality of $\textbf{a}$ implies that $\textbf{y} \prec_{lex} \textbf{f}_{r_u + 1}$. Since $\textbf{y} \in N(r)$, it follows that $\textbf{f}_{r_u + 1} \in N(r)$ and hence $\textbf{f}_{r_u + 1} \in N_u$. This contradicts $|N_u| = r_u$. This completes the proof. 
\end{proof}

\begin{lemma}\label{claim2}
With notations as in Lemma \ref{claim} and $u_2 < u_1 - 1,$ we have $$|\nabla (N(r))| = r - |N_{u_1}| + |\nabla (N_{u_1})|.$$
\end{lemma}

\begin{proof}
It follows from Lemma \ref{claim} that,
\begin{equation}\label{bigcup}
\bigcup_{u_2 < u \le u_1} \nabla_{u_1} (N_u) \subset N_{u_1}.
\end{equation}
This implies,
\begin{align*}
|\nabla(N(r))| &= |\nabla (N(r)) \cap F_{<u_1}| +  |\nabla (N(r)) \cap F_{\ge u_1}| \\
&=  |\nabla (N(r) \setminus N_{u_1}) \cap F_{<u_1}| + |\nabla (N_{u_1})|.
\end{align*}
Note that, $N(r) \setminus N_{u_1}$ consists of the first $r - |N_{u_1}|$ elements of $F_{u_2}^{u_1 - 1}$ in descending lexicographic order. 
We obtain by applying \eqref{bigcup} to $N(r) \setminus N_{u_1}$ (on $F_{u_2}^{u_1 - 1}$) that $\nabla_{u_1-1} (N(r) \setminus N(u_1)) \subset N_{u_1-1}$. Also, $N_{u_1 - 1} \subset N(r) \setminus N(u_1)$. This implies that $\nabla_{u_1-1} (N(r) \setminus N(u_1)) = N_{u_1-1}$. Repeating the argument iteratively we deduce that,
$$\nabla_{u} (N(r) \setminus N(u_2)) = N_{u} \ \ \ \text{for all} \ \ \ u_2 < u \le u_1-1.$$ 
Consequently, $\nabla (N(r) \setminus N_{u_1}) \cap F_{<u_1} = N(r) \setminus N_{u_1}$, which  proves the lemma.
\end{proof}

We are now ready to state and prove the main theorem of this section. This is a generalization of \cite[Theorem 3.8]{BD}.
Further special cases, when $d_1 = \dots = d_m = q$, appear as \cite[Lemma 6]{W}, \cite[Theorem 5.7]{HP} and \cite[Lemma 4.6]{GM}.

\begin{theorem}\label{w}
Let $u_1, u_2, u, r$ be integers with $-1 \le u_2 < u \le u_1 \le k$ and let $S \subseteq F_{u_2}^{u_1}$ with $|S| = r$. Then $|\nabla  (N (r))| \le |\nabla (S)|$. In particular, given any $S \in \mathcal{F}_r$, we have $|\Delta (S)| \le |\Delta(N(r))|$. Consequently, $$|\Delta(N(r)| = \max \{|\Delta(S)| : S \in \mathcal{F}_r \}.$$
\end{theorem}

\begin{proof}
For $u_2 < u \le u_1$, define $S_u := S \cap F_u$ and $N_u := N(r) \cap F_u$. When $u_2 = u_1 - 1$, then the assertion follows directly from Theorem \ref{CLT}. Henceforth, we will always assume that $u_2 < u_1 - 1.$ We distinguish the proof in two cases:

\textbf{Case 1:}  Suppose that $|S_{u_1}| \ge r_{u_1}$. Then $|S_{u_1}| = r_{u_1} + \alpha$ for some $\alpha \ge 0$. We may write $S_{u_1} = S' \cup S''$, where $S'$ denotes the first $r_{u_1}$ elements of $S$ in descending lexicographic order and $S'' = S \setminus S'$. It follows easily that $|S''| = \alpha$ and that $S''$ is disjoint from $\nabla(S)$ and $\nabla (N_{u_1})$. By applying Corollary \ref{cl} (b) to $S'$, we see that $|\nabla (S')| \ge |\nabla (N_{u_1})|$. This shows that
$|\nabla (S_{u_1})| \ge |\nabla (N_{u_1})| + \alpha$. We note that,

\begin{align}
|\nabla (S)| &= |\nabla_{<u_1} (S)| + |\nabla_{\ge u_1}(S)| \nonumber\\
&\ge |\nabla_{<u_1} (S)| + |\nabla (S_{u_1})|  \nonumber\\ 
&\ge |S \cap F_{<u_1}| + |\nabla (S_{u_1})| \nonumber\\
& = r - |S_{u_1}| + |\nabla (S_{u_1})|.   \label{1}
\end{align}

This gives
$$|\nabla (S)|  = r - |S_{u_1}| + |\nabla (S_{u_1})| 
\ge r - r_{u_1} - \alpha + |\nabla (N_{u_1})| + \alpha 
= |\nabla (N(S))|.$$ 

\smallskip
The last equality follows from Lemma \ref{claim2} and the proof is complete in this case. 
\textbf{Case 2:}
Now suppose that $|S_{u_1}|< r_{u_1}$. Since $|S|=r=|N(r)|$, there exists and integer $u$ with $u_2 < u <{u_1}$ such that $|S_u|>|N_u|$ and consequently, $|N_u^*| \le |S_u|$. By Lemma \ref{claim} and Corollary \ref{cl} (a) we have $|N_{u_1}| \le |\nabla_{u_1}(N_u^*)| \le |\nabla_{u_1}(S_u)|.$ Thus,
\begin{align*}
|\nabla(S)| &\ge r-|S_{u_1}|+|\nabla_{\ge {u_1}}(S)| \ \ \ \ \ (\text{follows from  \eqref{1}})\\
&> r - |N_{u_1}|+|\nabla_{\ge {u_1}}(S_u)| \\
&= r - |N_{u_1}|+|\nabla(\nabla_{u_1}(S_u))| \\
&\ge r - |N_{u_1}|+|\nabla(N_{u_1})|=|\nabla(N(S))|.
\end{align*}
The last equality follows from Lemma \ref{claim2}. The last two assertions are now obvious. 
\end{proof}
In order to answer Question \ref{q3} we must now determine $|\nabla (N(r))|$. To proceed we will need the following Lemma that was proved in \cite[Lemma 4.2]{BD}.

\begin{lemma}\cite[Lemma 4.2]{BD}\label{fp}
Let $d >0$ be an integer and $\textbf{a}_1, \dots, \textbf{a}_r$ be the first $r$ elements of $F_{\le d}$ in descending lexicographic order. Then,
$$\nabla (\textbf{a}_1, \dots, \textbf{a}_r) = \{\textbf{a} \in F : \textbf{a}_r \le_{lex} \textbf{a}\}.$$
Moreover, if $\textbf{a}_r = (a_{r, 1}, \dots, a_{r, m})$ then 
$$|\nabla (\textbf{a}_1, \dots, \textbf{a}_r)| = d_1 \cdots d_m  - \displaystyle{\sum_{i=1}^m a_{r, i} \prod_{j=i+1}^m d_j}.$$
\end{lemma}

The following Proposition, where we compute the $|\nabla (N(r))|$ completes our pursuit of answering Question \ref{q3}.

\begin{proposition}\label{card}
Let $u_1, u_2, u, r$ be as before. Assume that $N(r) := \{\textbf{a}_1, \dots, \textbf{a}_r\}$. Suppose $\textbf{a}_r$ is the $s$-th element of $F_{\le u_1}$ in descending lexicographic order. Then, $$|\nabla (\textbf{a}_1, \dots, \textbf{a}_r)| = d_1 \cdots d_m  - \displaystyle{\sum_{i=1}^m a_{r, i} \prod_{j=i+1}^m d_j} -s + r.$$
\end{proposition}

\begin{proof}
Let us denote by $M_{u_1} (s)$ the first $s$ elements of $F_{\le u_1}$ in descending lexicographic order. Clearly, $\textbf{a}_i \in M_{u_1} (s)$ for $i=1, \dots, r$. It is easy to see that 
$$\nabla (\textbf{a}_1, \dots, \textbf{a}_r) = \nabla (M_{u_1} (s)) \setminus \left(M_{u_1} (s) \setminus N(r)\right),$$
which proves that $$|\nabla (\textbf{a}_1, \dots, \textbf{a}_r)| = |\nabla (M_{u_1} (s))| - (s-r).$$
The assertion now follows from Lemma \ref{fp} by noting that $\textbf{a}_r$ is the $s$-th element of $M_{u_1} (s)$ in descending lexicographic order.
\end{proof}

We have thus answered the Question \ref{q3} completely and we note it down as the following corollary.

\begin{corollary}\label{combfinal}
Fix integers $u_1, u_2$ and $r$ with $-1 \le u_2 < u_1 \le k$. Denote by $\mathcal{F}_r$, the family of subsets of $F_{u_2}^{u_1}$ of cardinality $r$. Then $$\max \{|\Delta(S)| : S \in \mathcal{F}_r\} = \displaystyle{\sum_{i=1}^m a_{r, i} \prod_{j=i+1}^m d_j} +s - r,$$
where $(a_{r, 1}, \dots, a_{r, m})$ is the $r$-th element of $F_{u_2}^{u_1}$ in descending lexicographic order. In particular, 
\begin{enumerate}
\item[(a)] $a_r (u_1, u_2, \A) \le \displaystyle{\sum_{i=1}^m a_{r, i} \prod_{j=i+1}^m d_j} +s - r $ and 
\item[(b)] $M_r(u_1, u_2) \ge d_1 \cdots d_m  - \displaystyle{\sum_{i=1}^m a_{r, i} \prod_{j=i+1}^m d_j} -s + r.$ 
\end{enumerate}
\end{corollary}

\begin{proof}
The first assertion follows from Theorem \ref{w} and Proposition \ref{card}. The assertion (a) follows from equation \eqref{eqmax} and we now derive (b) as a consequence of equation \eqref{M}.
\end{proof}

In the following and the last section of this article, we will produce a set $\{f_1, \dots, f_r \} \in \mathcal{C}_r$ for which the upper bound for $a_r (u_1, u_2, \A)$ given in the Corollary \ref{combfinal} is attained.

\section{Maximal family of polynomials and the relative generalized Hamming weights of affine Cartesian codes}

As mentioned before, we now construct a family of polynomials $\{f_1,\dots,f_r\} \in \mathcal{C}_r$ such that $|\Z_\A(f_1,\dots,f_r)|$ attains the upper bound for $a_r (u_1, u_2, \A)$ as obtained in Corollary \ref{combfinal}. We call such a family of polynomials as a maximal family. First, recall that,  $A_i=\{\gamma_{i,1},\dots,\gamma_{i,d_i}\}$ for $i=1, \dots, m$.

\begin{definition}\label{def:fa}
For $\textbf{b} = (b_1, \dots, b_m) \in F$ define the polynomial,
$$f_{\textbf{b}} = \prod_{i=1}^m \prod_{j=1}^{b_i}(x_i - \gamma_{i, j}).$$
\end{definition}
We may note that $\deg f_{\textbf{b}} = b_1 + \dots + b_m$ and with respect to the graded lexicographic order the leading term of$f_{\textbf{b}}$ is given by $\LT (f_{\textbf{b}}) = x_1^{b_1} \cdots x_m^{b_m}$. We further observe that,
We define a map $\psi : \mathcal{A} \to F$ given by $(\gamma_{1, i_1}, \dots, \gamma_{m, i_m}) \mapsto (i_1 - 1, \dots, i_m - 1)$. The map $\psi$ is a bijection. It follows easily that for $\gamma \in \A$,
\begin{equation}\label{nv}
f_{\textbf{b}} (\gamma) \neq 0 \iff  \psi(\gamma)  \in \nabla(\textbf{b})
\end{equation}
We have the following proposition which is an analogue of \cite[Proposition 4.5]{BD}.
\begin{proposition}\label{maximal}
Let $\textbf{a}_1, \dots, \textbf{a}_r$ be the first $r$ elements of $F_{u_2}^{u_1}$ in descending lexicographic order and suppose that $\textbf{a}_r$ is the $s$-th element of $F_{\le u_1}$ in descending lexicographic order. 
Then, $$|{\rm Supp} (f_{\textbf{a}_1}, \dots, f_{\textbf{a}_r})| = d_1 \cdots d_m - \sum_{i=1}^m a_{r,i} \prod_{j=i+1}^m d_j - s + r,$$ where
$\textbf{a}_r = (a_{r,1}, \dots, a_{r, m})$ and ${\rm Supp} (f_{\textbf{a}_1}, \dots, f_{\textbf{a}_r}) = \mathcal{A} \setminus \Z_\A(f_{\textbf{a}_1}, \dots, f_{\textbf{a}_r})$.
\end{proposition}

\begin{proof}
It follows from equation \eqref{nv} that $\gamma \in {\rm Supp}(f_{\textbf{a}_1}, \dots, f_{\textbf{a}_r})$ if and only if $\psi (\gamma) \in \Delta (\textbf{a}_1, \dots, \textbf{a}_r)$. 
Thus, $|{\rm Supp}(f_{\textbf{a}_1}, \dots, f_{\textbf{a}_r})| = |\nabla (\textbf{a}_1, \dots, \textbf{a}_r)|$.  
Since $\textbf{a}_1, \dots, \textbf{a}_r$ are the first $r$ elements of $F_{u_1}^{u_2}$ in descending lexicographic order we see that, $$|{\rm Supp}(f_{\textbf{a}_1}, \dots, f_{\textbf{a}_r})| = |\nabla (\textbf{a}_1, \dots, \textbf{a}_r)| =  d_1 \cdots d_m  - \displaystyle{\sum_{i=1}^m a_{r, i} \prod_{j=i+1}^m d_j} - s + r,$$ where the last equality follows from Proposition \ref{card}. This completes the proof. 
\end{proof}

Finally we may state the main result of this paper where we compute all the RGHWs of an affine Cartesian code with respect to a smaller affine Cartesian code.

\begin{theorem}\label{last}
Fix integers $u_1, u_2$ with $-1 \le u_2 < u_1 \le \sum_{i=1}^m (d_i -1)$. Let $AC_q (u_1, \A)$ and $AC_q (u_2, \A)$ denote the corresponding affine Cartesian codes. For any integer $1 \le r \le \ell:= \dim AC_q (u_1, \A) - \dim AC_q(u_2, \A)$, the $r$-th RGHW of $AC_q (u_1, \A)$ with respect to $AC_q (u_2, \A)$, denoted by $M_r (u_1, u_2)$ is given by
$$M_r (u_1, u_2) =  d_1 \cdots d_m  - \displaystyle{\sum_{i=1}^m a_{r, i} \prod_{j=i+1}^m d_j} -s + r,$$
where $(a_{r,1}, \dots, a_{r,m})$ is the $r$-th element of $F_{u_2}^{u_1}$ and $s$-th element of $F_{\le u_1}$ in descending lexicographic order.
\end{theorem}

\begin{proof}
The result follows from Corollary \ref{combfinal} and Proposition \ref{maximal}.
\end{proof}

\section{Acknowledgments}
The author expresses his gratitude to Olav Geil for pointing out this problem,  Peter Beelen for some enlightening discussions on these topics, and Trygve Johnsen for his careful reading of the manuscript and several comments.


\begin{thebibliography}{GV1}
\bibitem{BD}
P. Beelen and M. Datta,
Generalized Hamming weights of affine cartesian codes,
\emph{Finite Fields Appl.} \textbf{51} (2018), 130 -- 145.

\bibitem{BDG}
P. Beelen, M. Datta and S. R.  Ghorpade, Vanishing ideals of projective spaces over finite fields and a projective footprint bound. \emph{Acta Math. Sin.} (Engl. Ser.) {\bf 35} (2019), no. 1, 47 -- 63.

\bibitem{C}
C.~Carvalho, On the second Hamming weight of some Reed-Muller type codes. \emph{Finite Fields Appl.} \textbf{24} (2013), 88 -- 94.

\bibitem{CN}
C.~Carvalho and V. G. L. Neumann,  On the next-to-minimal weight of affine cartesian codes. \emph{Finite Fields Appl.} \textbf{44} (2017), 113 -- 134.

\bibitem{CL}
G. F. Clements and B. Lindstr\"{o}m,
{A generalization of a combinatorial theorem of Macaulay},
\emph{J. Combinatorial Theory} {\bf 7}, 1969, 230 -- 238.

\bibitem{CLO}
D. A. Cox, J. Little and D. O'Shea,
Ideals, varieties, and algorithms.
An introduction to computational algebraic geometry and commutative algebra. Fourth edition. Undergraduate Texts in Mathematics. Springer, Cham, 2015.

\bibitem{CLO2}
D. A. Cox, J. Little and D. O'Shea,
Using Algebraic Geometry. Second edition. Graduate Texts in Mathematics. Springer, New York, 2005.

\bibitem{FL}
J. Fitzgerald and R. F. Lax,  Decoding affine variety codes using Gröbner bases. \emph{Des. Codes Cryptogr.}, {\bf 13}, 147 -- 158 (1998)

\bibitem{G}
O. Geil, Evaluation codes from an affine variety code perspective, Chapter 2 in Advances in Algebraic Geometry Codes, Series on Coding Theory and Cryptology, vol.5, World Scientific Publishing Co. Pte. Ltd., 2008.

\bibitem{GH}
O. Geil and T. Høholdt, Footprints or generalized Bezout's theorem. \emph{IEEE Trans. Inform. Theory}, {\bf 46}, 635--641 (2000)

\bibitem{GM}
O. Geil and S. Martin,
Relative generalized Hamming weights of $q$-ary Reed-Muller codes,
\emph{Adv. Math. Commun.}, {\bf 11}, No. 3, 2017,  503 -- 531.  

\bibitem{GT1}
O. Geil and C. Thomsen, 
Weighted Reed-Muller codes revisited, 
\emph{Des. Codes. Crypt}, 66, (2013),   195 -- 220.

 
\bibitem{GMVV}
M. Gonzalez-Sarabia, J. Martínez-Bernal, R. H. Villarreal and C. E. Vivares,
Generalized minimum distance functions,
arXiv:1707.03285. 

\bibitem{H}
T. Høholdt,  On (or in) Dick Blahut's footprint, \emph{ Codes, Curves and Signals}, pp. 3 -- 9, Kluwer, Norwell, MA, 1998

\bibitem{HP}
P. Heijnen and R. Pellikaan,  {Generalized Hamming weights of $q$-ary Reed-Muller codes}, \emph{IEEE Trans. Inform. Theory} {\bfseries 44} (1998), 181--196.

 

\bibitem{LCL}
Z. Liu, W. Chen and Y. Luo, The relative generalized Hamming weight of linear q-ary codes and their subcodes, Des. Codes Crypt., {\bf 48} (2008), 111 -- 123.

\bibitem{LRV}
H.H. L\'{o}pez, C. Renter\'{i}a-M\'{a}rquez and R.H. Villarreal,
Affine cartesian codes,
\emph{Des. Codes Cryptogr.} {\bf 71} (1) (2014), 5 --19.


\bibitem{Lea}
Y. Luo, C. Mitrpant, A. H. Vinck and K. Chen, 
Some new characters on the wire-tap channel of type II, \emph{IEEE Trans. Inf. Theory}, {\bf 51} (2005), 1222 -- 1229.

\bibitem{NPV}
 L. N\'u\~nez-Betancourt, Y. Pitones and R. H. Villarreal, Footprint and minimum distance functions, \emph{Commun.
Korean Math. Soc.} {\bf 33} (2018), no. 1, 85 -- 101.

\bibitem{NW}
Z. Nie, A. Y. Wang,
Hilbert functions and the finite degree Zariski closure in finite field combinatorial geometry,
\emph{J. Combin. Theory Ser. A} {\bf 134} (2015), 196 -- 220.

\bibitem{W}
V. K. Wei,  Generalized Hamming weights for linear codes,
\emph{IEEE Trans. Inform. Theory} {\bfseries 37} (1991), 1412--1418.


\end{thebibliography}
\end{document}